\documentclass[12pt,A4]{article}
\usepackage{amsmath}
\usepackage{esvect}
\usepackage{amsthm}
\usepackage{amssymb}
\usepackage{graphicx}
\usepackage{color}
\usepackage{lipsum}
\usepackage{caption}
\usepackage{changepage}
\usepackage{graphicx}
\usepackage{pdfpages}
\newtheorem{theorem}{Theorem}[section]

\newtheorem{lemma}[theorem]{Lemma}
\newtheorem*{remark}{Remark}
\theoremstyle{definition}

\newcommand{\forceindent}{\leavevmode{\parindent=1em\indent}}

\date{}
\usepackage[margin=0.8in]{geometry}

\begin{document}
\title{Static Stability Analysis of a Thin Plate with a Fixed Trailing Edge in Axial Subsonic Flow: Possio Integral Equation Approach}
\author{
Mohamed  Serry\\
Amjad Tuffaha}
\date{}
\maketitle

\bigskip

\indent Department of Mathematics and Statistics\\
\indent American University of Sharjah\\
\indent Sharjah, UAE\\
\indent e-mail:  mohamedserry91\char'100gmail.com \\
\indent e-mail: atufaha\char'100aus.edu

\maketitle		

\begin{abstract}
In this work, the static stability of  plates with fixed trailing edges in axial airflow is studied using the framework of Possio integral equation. First, we introduce a new derivation of a  Possio integral equation that relates the pressure jump along thin plates to their downwash based on the linearization of the governing equations of an ideal   compressible fluid. The steady state solution to the Possio equation is used to account for the aerodynamic forces in the steady state plate governing equation resulting in a singular differential-integral equation which is transformed to an integral equation. Next, we verify the solvability of the integral equation based on the Fredholm alternative for compact operators in Banach spaces and the contraction mapping theorem. Then, we derive explicit formulas for the characteristic equations of free-clamped and free-pinned plates. The minimum solutions to the characteristic equations are the divergence speeds which indicate when static instabilities start to occur. We show analytically that free-pinned plates are statically unstable. After that, we move to derive analytically  flow speed intervals  that correspond to static stability regions for free-clamped plates. We also resort to numerical computations to obtain an explicit formula for the divergence speed of free-clamped plates. Finally, we apply the obtained results on piezoelectric plates and we show that free-clamped piezoelectric plates are statically more stable than conventional free-clamped plates due to the piezoelectric coupling.    
\end{abstract}

\section{Introduction}
\forceindent\textit{Aeroelasticity} is a classical subfield of fluid mechanics that is concerned with the interactions between air flow and elastic bodies. Such interactions can have gentle effects such as flag flapping or may result in catastrophic consequences such as the collapse of the Tacoma Narrows Bridge in 1940. Therefore, aeroelasticity is essential in many serious applications such as the design of airplanes, bridges, tall buildings, and so on, to insure  static and dynamic stabilities. Additionally, there has been a recent interest in exploiting aeroelastic instabilities for the purpose of energy harvesting \cite{DSMD}.

 A conventional aeroelastic analysis (see for example \cite{Hodges}) aims to find the flow speeds at which dynamic instabilities (flutter) or static instabilities (divergence) of elastic structures  start to occur. Informally speaking, the minimum speeds at which static and dynamic instabilities start to occur are referred to as \textit{divergence speed} and \textit{flutter speed} respectively. In the design of systems such as airplanes and bridges, we aim to delay the divergence and flutter speeds so that these systems can operate over a wide range of flow speeds without triggering  static or dynamic instabilities. On the other hand, for energy harvesting aeroelastic systems, we seek to minimize the flutter speed as we can harvest more energy from the aeroelastic system when flutter occurs.

Due to the complexity of aeroelastic problems in general,  aeroelastic studies utilize rigorous numerical, experimental, and analytical treatments to insure thorough and accurate understanding of different aeroelastic phenomena. Despite their sophistication and limitations to simple problems, analytical techniques have contributed to the development of the field of aeroelasticity and understanding its aspects thoroughly. A very important example illustrating the effectiveness of analytical techniques in aeroelasticity is the outstanding  work of Theodorsen in the 1950's \cite{Theodorsen} who used tools from complex analysis to derive formulas of the aerodynamic loads on thin airfoils in incompressible airflow. Until now, a significant number of  scientists after Theodorsen have been using his formulas to study different aeroelastic problems (see for example \cite{aeroelastic control1, aeroelastic control2, aeroelastic control3, aeroelastic control4}).

Another important example that is usually overlooked is the interesting work of A.V. Balakrishnan. Balakrishnan implemented rigorous mathematical tools to derive and solve a singular integral equations (known as Possio equations) from which the aerodynamic loads on thin structures in compressible potential flows can be obtained \cite{Balakrishnan solving possio equation}. Balakrishnan, again equipped with rigorous mathematical tools and minimal amount of numerical computations, moved to conducting static and dynamic aeroelastic  analyses on continuum wing structures in normal subsonic flows \cite{iLiff paper,Transonic Dip}. 
In \cite{Axial flow conference paper}, Balakrishnan set a framework for  the analysis of the steady state or transient responses of thin plate in axial flow.

The axial flow over thin plates changes the nature of the fluid-structure interaction, as the plates' deflections vary nonlinearly along  the direction of the flow, and presents new mathematical and computational challenges. Moreover, understanding the axial flow problem has a very promising application in harvesting energy from winds. Therefore, there has been a series of recent mathematical, numerical, and experimental research works on the axial flow over thin plates and its energy harvesting applications  \cite{S,H,GWD, MD}.

A quick look at the literature on the axial flow over thin plates shows that most of the conducted studies consider the dynamic responses and instabilities. On the other hand, the number of  research papers on the static stabilities of plates in axial flow is minute (see for example \cite{static instability 1,static stability 2}).

Motivated and inspired by the work of Balakrishnan, the significance and challenges in analyzing axial flows, and the limited research on the static instabilities of plates in axial flow, we propose in this paper a framework to study the static instabilities of thin plates in axial air flow. We formulate the static aeroelastic equation of a thin plate based on the steady state solution to a Possio integral equation. The Possio equation is derived based on a linearization of the flow equations of an ideal compressible fluid. We then verify the existence and uniqueness of solutions to the static aeroelastic equation without a consideration of the boundary conditions and that makes our framework applicable for different axial flow problems with different boundary conditions. We derive characteristic equations explicitly for the  cases of free-pinned and free-clamped plates. The minimum solutions to the characteristic equations are the divergence speeds. After that, we analyze and solve the characteristic equations analytically and numerically. We show that the divergence speed for free-pinned plates is zero and that indicates that free-pinned plates are statically unstable. We also obtain analytically a flow velocity interval that guarantees static stability of free-clamped plates. Moreover, we derive an explicit formula for the divergence speed of free-clamped plates based on a numerical solution to the characteristic equation. Finally, we apply the previous results on piezoelectric plates and we show that free-clamped piezoelectric plates are statically more stable than conventional free-clamped plates due to piezoelectric coupling.

\section{Preliminaries}
In this section, we state some mathematical definitions and results that will be used throughout this work. We refer to any standard functional analysis book to have a thorough understanding of the stated results and definitions. The symbols $\mathbb{N}$, $\mathbb{R}$, and $\mathbb{C}$ denote the real and complex numbers respectively. 
\begin{itemize}
\item
The Banach space $L^{p}[-b,b]$ with $p\geq 1$ is the space of functions $f:[-b,b] \rightarrow \mathbb{R}$ satisfying the property $\int_{-b}^{b}|f(x)|^{p}dx<\infty$. The notation $L^{p_0+}[-b,b]$ indicates any space $L^{p}[-b,b]$ with $p>p_0$ while $L^{p_0-}[-b,b]$ indicates any space $L^{p}[-b,b]$ with $1\leq p<p_0$. $L^{p}[-b,b]$ is equipped with the norm $||*||_{L^{p}[-b,b]}$  defined as  $||f||_{L^{p}[-b,b]}=\left(\int_{-b}^{b}|f(x)|^{p}dx\right)^{1/p}~,f \in L^{p}[-b,b]$. 

\item
Let $f\in L^{p}[-b,b]$ and $g\in L^{q}[-b,b]$ such that $1/p+1/q=1$ and $p,q\geq 1$. Then, H\" older's inequality states that $\int_{-b}^{b}|fg|dx\leq ||f||_{L^{p}}||g||_{L^{q}}$.

\item The Banach space $C[-b,b]$ is the space of continuous functions $f:[-b,b] \rightarrow \mathbb{R}$. $C[-b,b]$ is equipped with the norm $||*||_{C[-b,b]}$  defined as
$||f||_{C[-b,b]}=\sup_{x\in[-b,b]}|f(x)|,~f \in C[-b,b]$. 

\item Let $\{f_n\}$ be a bounded sequence in $C[-b,b]$ that is equi-continuous (That means the value $|f_{n}(y)-f_{n}(x)|$ can be set  arbitrarily small by only setting the value of $|x-y|$ sufficiently small with a value independent of $n$). Then, by  Arzela-Ascoli theorem, there exists a subsequence of $\{f_n\}$ that is convergent.
\item A function $f:R\subset \mathbb{R}^n\rightarrow \mathbb{R}$ is called uniformly continuous if, informally speaking, the value $|f(x)-f(y)|$ can be set arbitrarily small by setting the value of $||x-y||$ sufficiently small independent of the value of $x$ and $y$.
\item
Let $X$, $Y$ and $Z$ be Banach spaces. Let $T:X\rightarrow Y$ be a linear operator. If there exists a constant $c$ such that, for all $x\in X$, $||T(x)||_{Y}\leq c ||x||_{X}$, then $T$ is called a bounded operator. The smallest value of $c$, such that the above inequality holds, is called the operator norm of $T$ and is denoted by $||T||$. If $||T||<1$, then $T$ is called a contraction mapping.

\item Let $T:X\rightarrow Y$ and $G:Y\rightarrow Z$ be bounded operators. Then, the operator composition $GT:X\rightarrow Z$ is bounded with an operator norm estimated by $||GT||\leq||G||~||T||$. 

\item Let $T:X\rightarrow X$ be a contraction mapping and $I:X \rightarrow X$ be the identity operator. Then, $I-T$ is invertible with an inverse given by $\sum_{n=0}^{\infty} T^n$ where $T^n$ is the composition of $T$ with itself $n$ times. 

\item Let $T:X\rightarrow Y$ be a linear operator. If for every bounded sequence $\{x_n\}$ in $X$, the sequence $\{T(x_n)\}$ has a convergent subsequence, then $T$ is called a compact operator. Note that every compact operator is bounded. 

\item
Let $T:X\rightarrow Y$ and $G:Y \rightarrow Z$ be bounded operators. If either $T$ or $G$ is compact, then the operator composition  $GT:X\rightarrow Z$ is compact. 

\item Let $T:X\rightarrow X$ be a linear operator, then the point spectrum set of $T$, denoted by $\sigma_{p}(T)$, is defined as the set of eigenvalues $\lambda \in \mathbb{C}$ of the operator $T$. 

 \item Let $T:X\rightarrow X$ be a compact operator, then  $\sigma(T)$ is at most countable (a set is countable if there exists a bijective mapping from that set to $\mathbb{N}$).
 
 \item The Fredholm alternative states that  if $T:X\rightarrow X$ is a compact operator then the  operator $\mu I-T$, where $I:X \rightarrow X$ is the identity operator and $\mu  \in \mathbb{C}/\{0
 \}$, is invertible whenever $\mu \notin \sigma_{p}(T)$.
\end{itemize}

\section{Problem Description}

\forceindent A slender plate with length $2b$  and width $l$ is placed in an axial compressible flow with a free stream velocity $U$ in the length direction of the plate (see figure \eqref{fig:plate configuration}).
 \begin{figure}
\includegraphics[width=5in ,keepaspectratio=true]{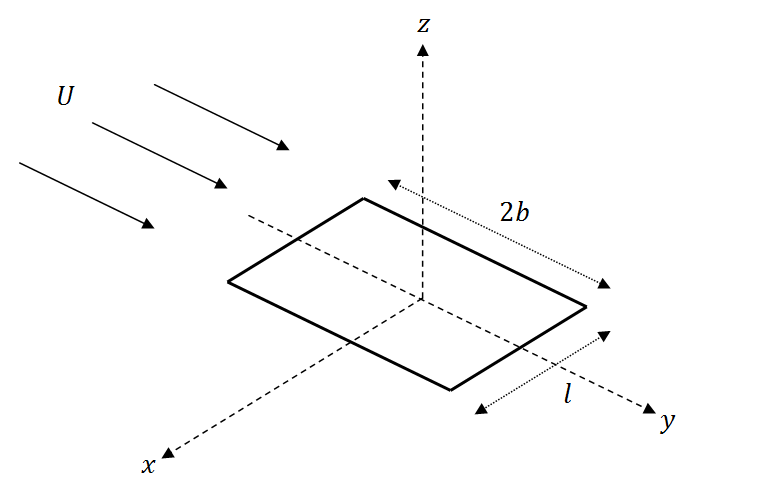}
\centering
\caption{plate configuration in an axial flow}
\label{fig:plate configuration}
\end{figure}
 If the plate leading edge is free and the trailing edge is pinned (no reaction moment), then we call that configuration free-pinned (see figure \eqref{fig:free pinned}). 
 \begin{figure}
\includegraphics[width=5in ,keepaspectratio=true]{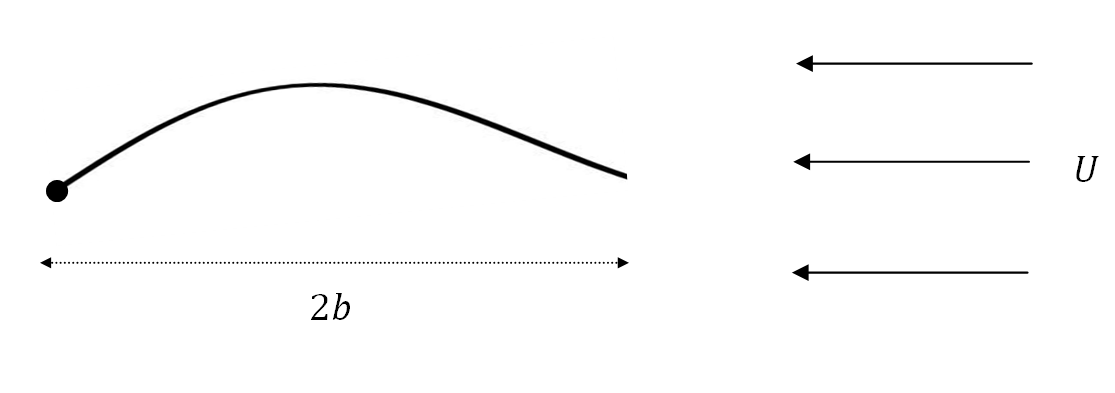}
\centering
\caption{free-pinned plate in an axial flow}
\label{fig:free pinned}
\end{figure}
On the other hand, if the leading edge is free and the trailing edge is clamped (there is a reaction moment), then we call that configuration free-clamped (see figure \eqref{fig:free clamped}).
\begin{figure}
\includegraphics[width=5in ,keepaspectratio=true]{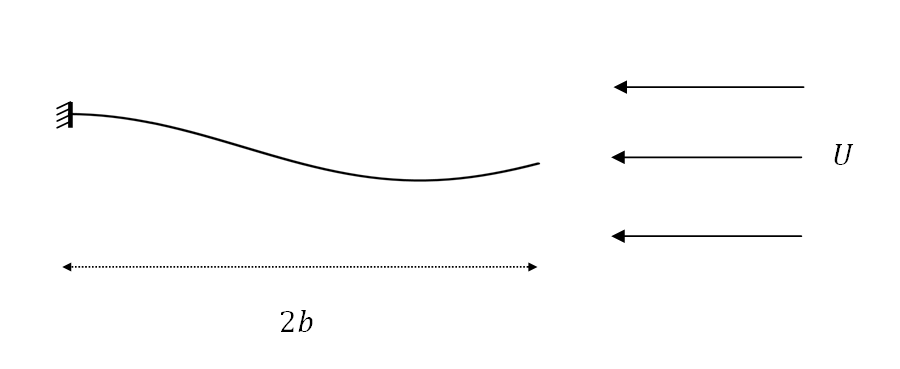}
\centering
\caption{free-clamped plate in an axial flow}
\label{fig:free clamped}
\end{figure}
We try to find, if it exists, an analytical or approximate formula for the divergence speed of the plate.     

\section{Plate Equation}
\forceindent In this work, we model the plate as an Euler-Bernoulli beam where the momentum balance is given by 
%
\begin{equation}
0=-\mathcal{M}_{yy}+\mathbb{F}_{fluid},
\end{equation}
where $\mathbb{F}_{fluid}$ is the fluid force per unit length, $ \mathcal{M}$ is the bending moment, and the subscript $(*)_y$  denote the derivatives with respect to $y$. The bending moment for Euler-Bernoulli beams is given by
\begin{equation}
\mathcal{M}=EI h_{yy},
\end{equation}
where $h(y)$ is the transverse deflection of the plate and $EI$ is the bending stiffness which is assumed to be constant along the length of the plate. Here, we  neglect the body forces and the tension along the length of the plate.  
Consequently, the plate is governed by the equation  
\begin{eqnarray}\label{eq:beam equation}
{EI}h_{yyyy} = \mathbb{F}_{fluid}&,&-b\leq y \leq b.
\end{eqnarray}
The boundary conditions for the  free-pinned plate are
\begin{equation}
h(b)=h_{yy}(b)=h_{yy}(-b)=h_{yyy}(-b)=0,
\end{equation}
and for the free-clamped plate, the boundary conditions are
\begin{equation}
h(b)=h_y(b)=h_{yy}(-b)=h_{yyy}(-b)=0.
\end{equation}
\section{Flow equations and the Possio integral equation}
\forceindent In this section, we derive and solve an integral equation, namely the Possio integral equation, based on the linearized Euler equation, linearized continuity equation, and linearized equation of state. The Possio integral equation relates the pressure jump along the plate to its downwash. 

First, we state the equations of an inviscid compressible two dimensional flow linearized about the free stream velocity $U$, the free stream flow pressure $p_0$ and the free stream density $\rho_0$. The flow equations are set to be two dimensional as it is assumed that the change in the flow variables in the direction of the width of the plate is negligible. This assumption is reasonable if we assume that $l\gg b$. The linearized Euler equations are
\begin{equation}\nonumber
\frac{\partial u}{\partial t}+U\frac{\partial u}{\partial y}=-\frac{1}{\rho_0}\frac{\partial p}{\partial y}
\end{equation}
and
\begin{equation}\nonumber
\frac{\partial v}{\partial t}+U\frac{\partial v}{\partial y}=-\frac{1}{\rho_0}\frac{\partial p}{\partial z}.
\end{equation}
The linearized continuity equation is given by
\begin{equation}\nonumber
\frac{\partial \rho}{\partial t}+U\frac{\partial \rho}{\partial y}+\rho_{0}\left(\frac{\partial u}{\partial y}+\frac{\partial v}{\partial z}\right)=0.
\end{equation}
Finally, the linearized state equation is give by 
\begin{equation}\nonumber
p=a_{\infty}^2\rho.
\end{equation}
The functions $u(y,z,t)$, $v(y,z,t)$, $p(y,z,t)$, and $\rho(y,z,t)$  are perturbation terms corresponding to  the flow  velocity component in the $y$ direction, the flow velocity component in the $z$ direction, the flow pressure, and the flow density respectively. The term $a_\infty$ is the free stream speed of sound which depends on the nature of the flow (for example: isothermal, isentropic, and so on).

The boundary conditions of the flow equations  are as follows. For any of the perturbation terms, denoted generically by $\delta (y,z,t)$, we have the far field condition
\begin{equation}\nonumber
\lim_{y^2+z^2\rightarrow \infty} \delta =0,
\end{equation}
and we also assume zero initial conditions for all perturbation terms.
Additionally, the pressure jump $\Delta p=p(y,0^+,t)-p(y,0^-,t)$ satisfies the Kutta-Joukowski condition
\begin{equation}\nonumber
\Delta p=0 ,~~|y|>b
\end{equation}
and the Kutta condition 
\begin{equation}\nonumber
\lim_{y\rightarrow b^{-}}\Delta p=0.
\end{equation}
Finally, the plate deformation is coupled with the flow by matching the normal velocities through the boundary condition
\begin{equation}\nonumber
v(y,0,t)=w_a,~~|y|\leq b,
\end{equation}
where $w_a$ is the downwash or the normal velocity on the plate surface. 

The derivation of the Possio equation starts with applying the Laplace transform in the $t$ variable and the Fourier transform in the $y$ variable  on the linearized equations to  result in the equations
\begin{equation}\label{eq:xmomentum transformed}
\lambda\hat{\hat{u}}+i\omega U\hat{\hat{u}}=(\lambda+i\omega U)\hat{\hat{u}}=-\frac{i\omega}{\rho_0}\hat{\hat{p}},
\end{equation}

\begin{equation}\label{eq:ymomentum transformed}
\lambda\hat{\hat{v}}+i\omega U\hat{\hat{v}}=(\lambda+i\omega U)\hat{\hat{v}}=-\frac{1}{\rho_0}\frac{\partial \hat{\hat{p}}}{\partial z},
\end{equation}

\begin{equation}\label{eq:continuity transformed}
\lambda\hat{\hat{\rho}}+i\omega U\hat{\hat{\rho}}+\rho_0\left(i\omega \hat{\hat{u}}+\frac{\partial \hat{\hat{v}}}{\partial z}\right)=0,
\end{equation}
and
\begin{equation}\label{eq:power law transformed}
\hat{\hat{p}}=a_{\infty}^2\hat{\hat{\rho}},
\end{equation}
where $\hat{f}(y,z,\lambda)=\int_{0}^{\infty}e^{-\lambda t}f(y,z,t) \, dt$ is the Laplace transform,  $Re(\lambda)\geq \sigma>0$, and $\hat{\hat{f}}(\omega,z,\lambda)=\int_{-\infty}^{\infty}e^{-i \omega x}\hat{f}(y,z,\lambda)\, dy$ is the Fourier transform. Substituting equation \eqref{eq:power law transformed} and \eqref{eq:xmomentum transformed} into equation \eqref{eq:continuity transformed} and solving for $\hat{\hat{p}}$ result in
\begin{equation}\label{eq:p in terms of v}
\hat{\hat{p}}=-\frac{\rho_0  a_\infty^2}{\lambda+i\omega U}\left(1+\frac{a_\infty ^2 \omega ^2}{(\lambda+i\omega U)^2}\right)^{-1}\frac{\partial \hat{\hat{v}}}{\partial z}.
\end{equation}
Substituting equation \eqref{eq:p in terms of v} into equation \eqref{eq:ymomentum transformed} results in

\begin{equation}\nonumber
(\lambda+i\omega U)\hat{\hat{v}}=\frac{ a_\infty^2}{\lambda+i\omega U}\left(1+\frac{ a_\infty ^2\omega ^2}{(\lambda+i\omega U)^2}\right)^{-1}\frac{\partial^2 \hat{\hat{v}}}{\partial z^2}
\end{equation}
or 
\begin{equation}\label{eq:DE for v}
\frac{\partial^2 \hat{\hat{v}}}{\partial z^2}=\left(\frac{(\lambda+i\omega U)^2}{a_\infty^2}+\omega ^2\right)\hat{\hat{v}}=B(\omega,\lambda)\hat{\hat{v}}.
\end{equation}
The solution to equation \eqref{eq:DE for v} is given by

\begin{equation}\label{eq:v hat hat}
\hat{\hat{v}}(\omega,z,\lambda)=\hat{\hat{v}}(\omega,0,\lambda)\begin{cases}
           e^{-\sqrt{B(\omega,\lambda)}z}, & z>0 ,\\           e^{\sqrt{B(\omega,\lambda)}z}, & z<0,
  \end{cases}
\end{equation}
where $\sqrt{*}$ is the square root with positive real part. Substituting the solution \eqref{eq:v hat hat} into equation \eqref{eq:ymomentum transformed} and integrating results in
\begin{equation}\label{eq:p hat hat}
\hat{\hat{p}}(\omega,z,\lambda)=-{\rho_0}\frac{(\lambda+i\omega U)}{\sqrt{B(\omega,\lambda)}}\hat{\hat{v}}(\omega,0,\lambda)\begin{cases}
           -e^{-\sqrt{B(\omega,\lambda)}z}, & z>0, \\
           e^{\sqrt{B(\omega,\lambda)}z}, & z<0.
         \end{cases}
\end{equation}
From equation \eqref{eq:p hat hat}, the pressure difference $\Delta p$ is given in the Fourier-Laplace domain by
\begin{equation}\nonumber
\Delta \hat{\hat{p}}={2\rho_0}\frac{(\lambda+i\omega U)}{\sqrt{B(\omega,\lambda)}}\hat{\hat{v}}(\omega,0,\lambda)
\end{equation}
and consequently we have, after using $\rho_0 U A=\Delta {p}$,
\begin{equation}\nonumber
2\hat{\hat{v}}(\omega,0,\lambda)=U\frac{\sqrt{B(\omega,\lambda)}}{\lambda+i \omega U}\hat{\hat{A}}.
\end{equation}
Using $k=\frac{\lambda}{U}$ in the above equation results in
\begin{equation}\label{eq:possio equation in the fourier domain}
2\hat{\hat{v}}(\omega,0,\lambda)=\frac{\sqrt{\tilde{B}(\omega,k)}}{k+i \omega }\hat{\hat{A}}.
\end{equation}
where 
$$\tilde{B}(\omega,k)=M^2(k+i\omega)^2+\omega^2.$$
Equation \eqref{eq:possio equation in the fourier domain} is referred to as the Possio equation in the Fourier domain. Note that the solutions of $\hat{\hat{v}}$ and $\hat{\hat{p}}$ given by \eqref{eq:v hat hat} and \eqref{eq:p hat hat} respectively are decaying as $B(\omega,\lambda)$ is never zero and therefore, the far field condition is satisfied.

We are interested in solving equation \eqref{eq:possio equation in the fourier domain} for the steady state case which corresponds to $\lambda=0$. Setting $\lambda=0$ reduces equation \eqref{eq:possio equation in the fourier domain} to
\begin{equation}\label{eq: Possio equation in the fourier domain k=0}
\frac{2}{\sqrt{1-M^2}}\hat{\hat{v}}(\omega,0,0)=\frac{|\omega|}{i \omega }\hat{\hat{A}}.
\end{equation}
The multiplier $|\omega|/i\omega$ corresponds to the Hilbert transform
\begin{equation}\nonumber
\mathcal{H}(f(t))(x)=\frac{1}{\pi}\int_{-\infty}^{\infty}\frac{f(t)}{x-t}\, dt.
\end{equation}
Therefore, equation \eqref{eq: Possio equation in the fourier domain k=0} corresponds to the integral equation (the variable t is dropped)
\begin{equation}
\frac{2}{\sqrt{1-M^2}}v(y,0)=\mathcal{H}(A)
\end{equation}
As the velocity of the flow is only known on the plate, we apply the projection operator (from now on, the term $p$ is not associated with the pressure) $\mathcal{P}:L^{p}(\infty,\infty)\rightarrow L^{p}[-b,b]$ on both sides of the above equation and use the Kutta-Joukowski condition which result in the Possio integral equation
\begin{equation}\label{eq:Possio integral equation}
\frac{2}{\sqrt{1-M^2}}w_a=\mathcal{H}_b(A),
\end{equation}
where $$\mathcal{H}_b(f(t))(x)=\frac{1}{\pi}\int_{-b}^{b}\frac{f(t)}{x-t}dt,~|x|\leq b$$ is the finite Hilbert operator.

The solvability of the Possio integral equation is illustrated as the following. In general, the solution to the Possio integral equation \eqref{eq:Possio integral equation} exists if $w_a \in L^{4/3+}[-b,b]$ and lies in $L^{4/3-}[-b,b]$ but the solution is not unique \cite{Tricomi}. If the Kutta condition is imposed and  $w_a\in L^{2+}[-b,b]$, then the solution to the Possio equation \eqref{eq:Possio integral equation} lies in $L^{4/3-}[-b,b]$ and is given uniquely by \cite{ground effect paper}
\begin{equation}\label{eq:solution to possio equation}
A=\frac{2}{\sqrt{1-M^2}}\mathcal{T}(w_a),
\end{equation}where 
\begin{equation}\nonumber
\mathcal{T}(f(t))(x)=\frac{1}{\pi}\sqrt{\frac{b-x}{b+x}}\int_{-b}^{b}\sqrt{\frac{b+t}{b-t}}\frac{f(t)}{t-x}dt,~|x|\leq b
\end{equation}
is the Tricomi operator.
\begin{remark}
In the work of Balakrishnan \cite{Bal}, the Possio equation was derived based on the linearization of the full nonlinear potential equation which assumes an ideal isentropic flow. Apparently, a linearization of the Euler, continuity, and state equations results in the same Possio equation that Balakrishnan derived. Despite the fact that we did not assume a potential (irrotational) flow in our derivation, the irrotationality comes from the linearization of the Euler equation about an irrotational velocity field. To illustrate this point, we write the linearized Euler equation in the vector form
\begin{equation}\label{eq:euler vector}
\frac{\partial \vv{u}}{\partial t}+\vv{U}\cdot \nabla \vv{u}=-\frac{1}{\rho_0} \nabla p,
\end{equation}
where $\vv{U}$ is the irrotational velocity field that the Euler equation is linearized about. Then, applying the curl operator $\nabla \times$ on each side of equation \eqref{eq:euler vector} results in
\begin{equation}\nonumber
\frac{\partial \Omega}{\partial t}+ U \frac{\partial}{\partial y} \Omega=0,
\end{equation}
where $Omega=\nabla \times \vv{u}$ is the 2D flow vorticity. This is a transport equation,  and  if we assume that the flow to be initially irrotational by imposing no perturbation initially in addition to the zero far field condition (in the y direction), then the flow will stay irrotational. 
Another way to show that the flow is irrotational is to apply the Fourier transform in the $y$ variable and the Laplace transform in the $t$ variable on the vorticity $\vv{\Omega}$. For our case of two dimensional flow, the vorticity has only one nonzero component in the $x$ direction which is given by $ \psi=\frac{\partial v}{\partial y}-\frac{\partial u}{\partial z}$. Therefore, using the solutions \eqref{eq:v hat hat}, \eqref{eq:p hat hat} and equation \eqref{eq:xmomentum transformed}, we have
\begin{equation}\label{eq:irrotational}
\hat{\hat{\psi}}=i\omega \hat{\hat{v}}-\frac{\partial \hat{\hat{u}}}{\partial z}=i\omega \hat{\hat{v}}-i\omega \hat{\hat{v}}=0,
\end{equation}
and that shows that the flow is irrotational. Equation \eqref{eq:irrotational} still holds for the steady state case.

\end{remark}

\section{Static aeroelastic equations}
\forceindent After obtaining a solution to the Possio integral equation, we have the aerodynamic force term $\mathbb{F}_{fluid}$ is given by 
$\mathbb{F}_{fluid}=l\Delta p$
assuming no change in the pressure along the width of the plate. For the steady state case, the downwash of the plate is given by $w_a=-Uh_y$. Therefore, using the solution \eqref{eq:solution to possio equation}, the plate governing equation can be written as
\begin{equation}\label{eq:final beam equation}
h_{yyyy}=W(h_y),	
\end{equation}
where the integral operator $W$ is defined as
\begin{equation}\nonumber
W(h_{y})=-\frac{2\rho_0 U^2l}{EI\sqrt{1-M^2}}\mathcal{T}(h_y).
\end{equation}
Equation \eqref{eq:final beam equation} is a singular differential-integral equation.  We need to find the values of $U$ such that this equation has solutions satisfying the plate boundary conditions. Such a problem can be referred to as an eigenvalue problem. Eigenvalue problems appear in many aeroelastic problems (for example: finding the flutter speed) and in engineering applications in general (for example: finding the buckling critical load of a beam).   

Now, we study the solvability of  equation \eqref{eq:final beam equation}  as the following. Let
\begin{equation}\nonumber
\mathbf{H}=\left( \begin{array}{c} h \\ {h}_y\\ {h}_{yy} \\ {h}_{yyy} \end{array} \right),
\end{equation}

\begin{equation}\nonumber
\mathcal{A}=\left(\begin{array}{cccc}0 & 1 & 0 & 0\\0 & 0 & 1 & 0\\0 & 0 &0 & 1\\0 & 0 & 0 & 0\end{array}\right),
\end{equation}
and
\begin{equation}\nonumber
\mathcal{W}=\left(\begin{array}{cccc}0 & 0 & 0 & 0\\0 & 0 & 0 & 0\\0 & 0 &0 & 0\\ 0& W & 0 & 0\end{array}\right).
\end{equation}
Then, the state space representation of equation \eqref{eq:final beam equation} is given by
\begin{equation}\label{eq: beam equation state space}
\mathbf{H}_{y}=\mathcal{A}\mathbf{H}+\mathcal{W}\mathbf{H}.
\end{equation}
The solution to equation \eqref{eq: beam equation state space} is equivalent to solving the integral equation (variation of parameters formula)
\begin{equation}\label{eq: equivalent integral equation}
\mathbf{H}(y)=e^{(y+b)\mathcal{A}}\mathbf{H}(-b)+\int_{-b}^{y}e^{(y-s)\mathcal{A}}\mathcal{W}\mathbf{H}ds,
\end{equation}
where the exponential matrix  $e^{y\mathcal{A}}$ is given by
\begin{equation}\nonumber
e^{y\mathcal{A}}=\left(\begin{array}{cccc}
\eta_1(y) & \eta_2(y) & \eta_3(y) & \eta_4(y)\\\eta_1'(y) & \eta_2'(y) & \eta_3'(y) & \eta_4'(y)\\\eta_1''(y) & \eta_2''(y) & \eta_3''(y) & \eta_4''(y)\\\eta_1'''(y) & \eta_2'''(y) & \eta_3'''(y) & \eta_4'''(y)
\end{array}\right)
=\left(\begin{array}{cccc}
1 & y & \frac{y^2}{2} & \frac{y^3}{6}\\0 & 1 & y & \frac{y^2}{2}\\0 & 0 & 1 & y\\0 & 0 & 0 & 1
\end{array}\right),
\end{equation}
where the primes denote the derivatives with respect to $y$.
Let
\begin{equation}\nonumber
[\mathcal{J}(\mathbf{H})](y)=\int_{-b}^{y}e^{(y-s)\mathcal{A}}\mathcal{W}\mathbf{H}ds.
\end{equation}
Then, equation \eqref{eq: equivalent integral equation} can be written in the abstract form
\begin{equation}\label{eq: abstract integral equation}
[(\mathcal{I}-\mathcal{J})\mathbf{H}](y)=e^{(y+b)\mathcal{A}}\mathbf{H}(-b),
\end{equation}
where $\mathcal{I}$ is the identity operator applied on $4\times 4$ matrices with entries in $L^{p}[-b,b],~p\geq 1$.
The operator $\mathcal{I}-\mathcal{J}$ is written explicitly as
\begin{equation}\nonumber
\mathcal{I}-\mathcal{J}=\left(\begin{array}{cccc} I &-T_1 & 0&0\\0&I-T_2&0&0\\0&-T_3& I&0\\0&-T_4&0&I\end{array}\right),
\end{equation} 
where $I:L^p[-b,b]\rightarrow L^{p}[-b,b]$ is the identity operator. The integral operators $T_i:L^{p}[-b,b]\rightarrow C[-b,b],i=1,2,3,4$ are defined as
$$[T_i(f)](y)=\int_{-b}^{y}\eta_{4}^{(i-1)}(y-s)W(f) \,ds,$$
where $\eta^{i}_4$ is the ith derivative of $\eta_4$. If $\mathcal{I}-\mathcal{J}$ is invertible, then the inversion formula is given by
\begin{equation}\nonumber
(\mathcal{I}-\mathcal{J})^{-1}=\left(\begin{array}{cccc} I &T_1(I-T_2)^{-1} & 0&0\\0&(I-T_2)^{-1}&0&0\\0&T_3(I-T_2)^{-1}& I&0\\0&T_4(I-T_2)^{-1}&0&I\end{array}\right).
\end{equation}It is noticed from the above formula that the operator $\mathcal{I}-\mathcal{J}$ is invertible if the operator $I-T_2$  is invertible. Analyzing the invertibility of the operator $I-T_2$ is equivalent to studying the integral equation
\begin{equation}\label{eq:integral equation to be solved}
(I-T_2)f=f(y)+\mu\int_{-b}^{y}(y-s)^2\mathcal{T}(f(t))(s)ds=g(y)~~,~~|y|\leq b,
\end{equation}
where $$\mu=\frac{\rho_0 U^2l}{EI\sqrt{1-M^2}}.$$
Note that the parameter $\mu$ will play an important role in the upcoming discussion. 

Now, we state some preliminary  lemmas, with their proofs, that are necessary for proving the solvability of equation \eqref{eq:integral equation to be solved}.

\begin{lemma}
The operator $\mathcal{L}:L^{p}[-b,b]\rightarrow C[-b,b],~p\geq1$ given by 
\begin{equation}\nonumber
[\mathcal{L}(f)](y)=\int_{-b}^{y}(y-s)^2 f(s) \,ds
\end{equation}
is compact.
\end{lemma}
\begin{proof}
First, we show that $\mathcal{L}$ is bounded. Let $f\in L^{p}[-b,b]$ and let $q$ be related to $p$ through the relation $1/p+1/q=1$, then using H\"older's inequality we have 
\begin{equation}\nonumber
\begin{split}
|\mathcal{L}(f)|& \leq (y+b)^2 \int_{-b}^{y}|f(s)|\, ds\\
& \leq 4b^2  \int_{-b}^{b}|f(s)|\, ds\\
& \leq 4b^2 (2b)^{1/q}||f||_{L^{p}[-b,b]}.
\end{split}
\end{equation}
Therefore $\mathcal{L}:L^{p}[-b,b]\rightarrow C[-b,b],p\geq1$ is bounded. Next, we show that the image of a bounded sequence $||f_n||_{L^{p}}\leq C_0,~C_0\geq 0$  under $\mathcal{L}$ is  equi-continuous. Assume $-b<x<y<b$ , then we have
\begin{equation}\nonumber
\begin{split}
|[\mathcal{L}(f_n)](y)-[\mathcal{L}(f_n)](x)|&=\left|\int_{x}^{y}(y-s)^2f_n (s)\, ds+\int_{-b}^{x}((y-s)^2-(x-s)^2)f_n (s)\,ds\right|\\
&\leq  \int_{x}^{y}(y-s)^2|f_n (s)|\,ds+\int_{-b}^{x}|(y-s)^2-(x-s)^2||f_n (s)|\,ds\\
&\leq (y-x)^2 \int_{-b}^{b}|f_n (s)|\, ds+\int_{-b}^{b}|(y-s)^2-(x-s)^2||f_n (s)|\, ds\\
& \leq  (y-x)^2 (2b)^{1/q}||f_n||_{L^{p}[-b,b]}+\left(\int_{-b}^{b}\left|(y-s)^2-(x-s)^2 \right|^{q}  \,ds \right)^{1/q}||f_n||_{L^{p}[-b,b]}
\end{split}
\end{equation}
The function $k(x,y)=(y-x)^2$ is uniformly continuous on $[-b,b]\times [-b,b]$ and additionally, $k(y,x)\rightarrow 0$  as $|y-x|\rightarrow 0$. Therefore, the right hand side of the above inequality can be set to be arbitrarily small for sufficiently small   $|y-x|$.  Then, by the Arzela-Ascoli theorem, there exists a convergent subsequence of the sequence $\mathcal{L}(f_n)$. Consequently, the operator $\mathcal{L}$ is compact.
\end{proof}

\begin{lemma}
The Tricomi operator is bounded from $C[-b,b]$ to $L^{4/3-}[-b,b]$.
\end{lemma}
\begin{proof}
It was shown in \cite{ground effect paper} that $\mathcal{T}:L^{2+}[-b,b]\rightarrow L^{4/3-}[-b,b]$ is bounded with an operator norm denoted by $||\mathcal{T}||$. Let $f\in C[-b,b]$, then
\begin{equation}\nonumber
||\mathcal{T}(f)||_{L^{4/3-}[-b,b]}\leq ||\mathcal{T}||~ ||f||_{L^{2+}[-b,b]}\leq (2b)^{1/2+}||\mathcal{T}||~ ||f||_{C[-b,b]} 
\end{equation}
and that completes the proof.

\end{proof}
\begin{theorem}
The integral equation \eqref{eq:integral equation to be solved} has a unique solution for a continuous, but not connected, range of values of $\mu$. 
\end{theorem}
\begin{proof}
 The operator $\mathcal{L}:L^{4/3-}[-b,b]\rightarrow C[-b,b]$ is compact and the operator $\mathcal{T}:C[-b,b]\rightarrow L^{4/3-}[-b,b]$ is bounded. Therefore, the operator $\mathcal{L}\mathcal{T}:C[-b,b]\rightarrow C[-b,b]$ is compact . Additionally, $\sigma_{p}(\mathcal{L}\mathcal{T})$  is at most countable as $\mathcal{L}\mathcal{T}$ is compact.  Then using the Fredholm Alternative, the integral equation  \eqref{eq:integral equation to be solved} has unique solutions whenever $\mu \neq-\frac{1}{\lambda}$ for all $\lambda \in \sigma(\mathcal{L}\mathcal{T})$. A weaker result can be obtained for small values of $\mu$. If  $\mu < \frac{1}{ ||\mathcal{L}||||\mathcal{T}||} $, then $T_2$ is a contraction mapping and therefore, the inverse of $I-T_2$ is given by $\sum_{n=0}^\infty T_2^n$ and that complete the proof.
\end{proof}
After we verified the invertibility of the operator $I-\mathcal{J}$ for a range of values of $\mu$, we have the solution to equation \eqref{eq:integral equation to be solved} is given by
\begin{equation}\label{eq:general solution to the integral equation}
\mathbf{H}(y)=[(\mathcal{I}-\mathcal{J})^{-1}e^{(\cdot+b)\mathcal{A}}\mathbf{H}(-b)](y).
\end{equation}
The previous theorem indicates that there exists  unique solution to the aeroelastic equations. However, this theorem and equation \eqref{eq:general solution to the integral equation} do not guarantee the existence of a solution to the aeroelastic equation that satisfies the boundary conditions of the plate. An analytical treatment to the existence of solutions satisfying the boundary conditions is outside the scope of this work. In the upcoming section, we derive the characteristic equations from which the divergence speed is obtained

\section{characteristic equations of  plates in axial flow}
\forceindent In this section, we derive the characteristic equations of the free-pinned and free-clamped plates. The minimum solutions to the characteristic equations are the divergence speeds. The derivation of the characteristic equations is obtained by matching the nonzero entries of $\mathbf{H}(-b)$ with  the zero terms of $\mathbf{H}(b)$ using the relation  

\begin{equation}\label{eq:Ax=0}
\mathbf{0}=\mathbf{P}[(I-\mathcal{J})^{-1}e^{(\cdot+b)\mathcal{A}}](b)\mathbf{Q}\mathbf{u},
\end{equation}
where
\begin{equation}\nonumber
\mathbf{P}=\left(\begin{array}{cccc}1&0&0&0\\0&0&1&0\end{array}\right)
\end{equation}
for free-pinned plates,

\begin{equation}\nonumber
\mathbf{P}=\left(\begin{array}{cccc}1&0&0&0\\0&1&0&0\end{array}\right)
\end{equation}
for free-clamped plates,

\begin{equation}\nonumber
\mathbf{Q}=\left(\begin{array}{c c} 1 & 0 \\ 0 & 1 \\ 0 & 0\\ 0 & 0
\end{array}\right)
\end{equation}
for both free-pinned and free-clamped plates, and 
\begin{equation}\nonumber
\mathbf{u}=\left(\begin{array}{c} h(-b)\\ {h}_{y}(-b) \end{array}\right)
\end{equation}
for both free-pinned and free-clamped plates. From equation \eqref{eq:Ax=0}, it is deduced that
\begin{equation}\label{eq:characteristic equation}
d(U)=\det \big(\mathbf{P}[(I-\mathcal{J})^{-1}e^{(\cdot+b)\mathcal{A}}](b)\mathbf{Q}\big)=0
\end{equation}
for which is the characteristic equation that we need solve. The smallest solution to \eqref{eq:characteristic equation}, if it exists, is the divergence speed $U_{div}$. Formally speaking, the divergence speed is defined as the minimum speed at which the static aeroelastic equations linearized about the steady state solution have a nonzero solution \cite{Transonic Dip}. In the following theorem, we show that the smallest solution to \eqref{eq:characteristic equation} satisfies the formal definition of the divergence speed.
\begin{theorem}
If there exists a solution $U$ to the characteristic equation \eqref{eq:characteristic equation}, then there exists a nonzero  solution to the aeroelastic equation \eqref{eq: equivalent integral equation} satisfying the plate boundary condition.
\end{theorem}
\begin{proof}
Let $U$ be a solution to the characteristic equation \eqref{eq:characteristic equation}. Therefore, the null space of the matrix $\mathbf{P}[(\mathcal{I}-\mathcal{J})^{-1}e^{( \cdot+b)\mathcal{A}}](b)\mathbf{Q}$ is nonzero. Let $\mathbf{u}$, defined previously, be a nonzero choice from the null space of the this matrix. Note that the null space is infinite, therefore, the constructed nonzero solution is not unique. By the continuity of the matrix $[(\mathcal{I}-\mathcal{J})^{-1}e^{(\cdot+b)\mathcal{A}}](y)$, the solution
$$\mathbf{H}(y)=[(\mathcal{I}-\mathcal{J})^{-1}e^{(\cdot+b)\mathcal{A}}](y)\mathbf{H}(-b)
$$
is nonzero and it satisfies the boundary conditions of the plate and that completes the proof.
\end{proof}

\forceindent By direct calculations, the characteristic equation for the free-pinned plates is explicitly given by
\begin{equation}\label{eq:characteristic equation free-pinned}
d(U)=[T_{3}(I-T_2)^{-1}(1)](b)=0.
\end{equation}
Moreover, the explicit formula for the free- clamped plates is explicitly given by
\begin{equation}\label{eq:characteristic equation explicit}
d(U)=[(I-T_2)^{-1}(1)](b)=0.
\end{equation}
\forceindent In the upcoming sections, we aim to analyze the derived characteristic equations and solve them either numerically or analytically.

\section{Static stability analysis of free-pinned plates}
\forceindent Here we state the main result directly. The main result of this section shows that free-pinned plates are statically unstable and that is illustrated through the following theorem.
\begin{theorem}
The minimum solution to the characteristic equation \eqref{eq:characteristic equation free-pinned} is $U=0$. 
\end{theorem}
\begin{proof}
If $U=0$, then $T_2=T_3=0$. Therefore,
\begin{equation}\nonumber
[(I-T_2)^{-1}(1)](y)=1
\end{equation}
and consequently,
\begin{equation}\nonumber
[T_3(I-T_2)^{-1}(1)](b)=0
\end{equation}
and that completes the proof.
\end{proof}

\section{Static stability analysis of free-clamped plates}
\forceindent In the following discussion, we aim to analyze the characteristic equation \eqref{eq:characteristic equation explicit} analytically and numerically.
\subsection{analytical study}
In this subsection, we show that there exists a stability range for $U$ such that the characteristic equation does not have a solution.
\begin{theorem}
 For $\mu\leq\frac{2\varepsilon}{5\pi b^3}$ with $0<\varepsilon<1$, there exists no solution to the characteristic equation \eqref{eq:characteristic equation explicit}.
 \end{theorem}
\begin{proof}
We assume that  
$$[(I-T_2)^{-1}(1)](b)=\sum_{n=0}^{\infty}T^{n}_2 (1)(b)$$ whenever $\sum_{n=0}^{\infty}[T^{n}_2 (1)](b)$ is convergent. We will verify the use of the above formula in the upcoming discussion. The operator $T_2$ can be written as
$$T_2=-\mu \mathcal{L}\mathcal{T}.$$
Consequently, the characteristic equation can be rewritten in terms of the parameter $\mu$ as
\begin{equation}\nonumber
S(\mu)=\sum_{n=0}^{\infty}(-1)^n c_n \mu^n=0,
\end{equation}
where
$$c_n=[\left(\mathcal{L}\mathcal{T}\right)^{n}(1)](b).$$

Next, we show that $c_n=[\left(\mathcal{L}\mathcal{T}\right)^n(1)](b)$ and $[\left(\mathcal{L}\mathcal{T}\right)^n(1)](y),~|y|\leq b$ are positive for all $n\geq 1$. For $n=1$, we have
\begin{equation}\nonumber
[\left(\mathcal{L}\mathcal{T}\right)(1)](y)=\int_{-b}^{y} (y-s)^2\sqrt{\frac{b-s}{b+s}}\,ds>0.
\end{equation}
Therefore, we have $[\left(\mathcal{L}\mathcal{T}\right)(1)](y), \left([\mathcal{L}\mathcal{T}\right)(1)](b)>0$ and by induction, if $[\left(\mathcal{L}\mathcal{T}\right)^{n}(1)](y),[\left(\mathcal{L}\mathcal{T}\right)^{n} (1)](b)>0$, we have
\begin{equation}\nonumber
\begin{split}
[\left(\mathcal{L}\mathcal{T}\right)^{n+1}(1)](b)&=\frac{1}{\pi}\int_{-b}^{b} (b-s)^2\sqrt{\frac{b-s}{b+s}}\int_{-b}^{b}\sqrt{\frac{b+t}{b-t}}\frac{[\left(\mathcal{L}\mathcal{T}\right)^{n}(1)](t)}{t-s}\, dt \, ds\\
&=\int_{-b}^{b} \mathcal{H}_b\left((b-s)^2\sqrt{\frac{b-s}{b+s}}\right)(t)\sqrt{\frac{b+t}{b-t}}[\left(\mathcal{L}\mathcal{T}\right)^{n} (1)](t)\, dt\\
&=\int_{-b}^{b} \left(\frac{1}{4}(2t-{3b})^2+\frac{5b^2}{4}\right)\sqrt{\frac{b+t}{b-t}}[\left(\mathcal{L}\mathcal{T}\right)^{n} (1)](t)\,dt>0.
\end{split}
\end{equation}
$[\left(\mathcal{L}\mathcal{T}\right)^{n+1}(1)](y)$ can be written as 
\begin{equation}\label{eq:LT(y)}
\begin{split}
[\left(\mathcal{L}\mathcal{T}\right)^{n+1}(1)](y)&=\frac{1}{\pi}\int_{-b}^{y} (y-s)^2\sqrt{\frac{b-s}{b+s}}\int_{-b}^{b}\sqrt{\frac{b+t}{b-t}}\frac{[\left(\mathcal{L}\mathcal{T}\right)^{n}(1)](t)}{t-s}\,dt\, ds\\
&=\frac{1}{\pi}\int_{-b}^{b} \int_{-b}^{y} \frac{(y-s)^2}{t-s}\sqrt{\frac{b-s}{b+s}}ds\sqrt{\frac{b+t}{b-t}}[\left(\mathcal{L}\mathcal{T}\right)^{n}(1)](t)\, dt>0\\
\end{split}
\end{equation}
 
as the integral 
\begin{equation}\nonumber
\int_{-b}^{y} \frac{(y-s)^2}{t-s}\sqrt{\frac{b-s}{b+s}}\, ds
\end{equation}
is given explicitly by
\begin{equation}\nonumber
\begin{split}
f(t,y)&=\sqrt{\frac{b-t}{b+t}}(y-t)^2\ln\left(\frac{\sqrt{\frac{b-y}{b+y}}+\sqrt{\frac{b-t}{b+t}}}{\sqrt{\frac{b-y}{b+y}}-\sqrt{\frac{b-t}{b+t}}}\right)\\
&+\left(2y^2+4(b-t)y+2t^2-2bt+b^2\right)\left(\frac{\pi}{2}-\arctan\left(\sqrt{\frac{b-y}{b+y}}\right)\right)\\
&+b\frac{(4y-2t+3b)\left(\frac{b-y}{b+y}\right)^{3/2}+(4y-2t+b)\sqrt{\frac{b-y}{b+y}}}{2\left(\frac{b-y}{b+y}\right)+\left(\frac{b-y}{b+y}\right)^2+1}
\end{split}
\end{equation}
which is positive for all $|t|\leq b$ and $|y|\leq b$. We verified that $c_n$ is positive and additionally, using \eqref{eq:LT(y)}, we have $[\left(\mathcal{L}\mathcal{T}\right)^{n}(1)](y)$ is an increasing positive function of $y$ with a maximum value $c_n$. Therefore, the series $S(\mu)$ is an alternating series. Next, we estimate a bound for the ratio $\frac{c_{n+1}\mu^{n+1}}{c_n\mu^{n}}$.

\begin{equation}\nonumber
\begin{split}
\frac{c_{n+1}\mu^{n+1}}{c_n\mu^{n}}&={\mu}\frac{c_{n+1}}{c_n}\\
&=\frac{\mu}{c_n}{\int_{-b}^{b} \left(\frac{1}{4}(2t-{3b})^2+\frac{5b^2}{4}\right)\sqrt{\frac{b+t}{b-t}}[T^n_2(1)](t)\, dt}\\
&\leq \mu{\int_{-b}^{b} \left(\frac{1}{4}(2t-{3b})^2+\frac{5b^2}{4}\right)\sqrt{\frac{b+t}{b-t}}\, dt}\\
\end{split}
\end{equation}
and from the above estimate, we impose that 
\begin{equation}\label{eq:bound on mu}
\mu \leq\frac{\varepsilon}{{\int_{-b}^{b} \left(\frac{1}{4}(2t-{3b})^2+\frac{5b^2}{4}\right)\sqrt{\frac{b+t}{b-t}}\, dt }}=\frac{2\varepsilon}{5\pi b^3},
\end{equation}
where $0<\varepsilon<1$. Consequently, we have that the sequence $c_n \mu^n$ is decreasing and approaching zero. Therefore, $S(\mu)$ is convergent by the alternating series test and additionally, it is absolutely convergent by the ratio test. In fact, with the above bound on $\mu$, we have $\sum_{n=0} ^\infty (-1)^n [(\mathcal{L}\mathcal{T})^n(1)](y)\mu^n$  is absolutely convergent for all $|y|\leq b$ using the direct comparison test with the series $\sum_{n=0}^\infty c_n \mu^n$ as we showed previously that $[(\mathcal{L}\mathcal{T})^n(1)](y)<c_n$. Therefore, the inversion formula $[(I-T_2)^{-1}(1)](y)=\sum_{n=0}^{\infty} [T^{n}_2(1)](y)$ is well-defined given that $\mu$ satisfies the bound  \eqref{eq:bound on mu}. The second term of the partial sum sequence of $S(\mu)$ is
\begin{equation}\nonumber
S_2(\mu)=1-\mu \frac{5\pi b^3}{2}>0
\end{equation}
due to the bound \eqref{eq:bound on mu}. Therefore, $S(\mu)>0$ as the term $c_n \mu^n$ is always decreasing by \eqref{eq:bound on mu}. 
\end{proof}
Based on the above theorem, and by assuming that $\varepsilon\rightarrow 1^{-}$ to maximize the convergence interval of $\mu$, we have the following static stability flow velocity range for the free-clamped plates. 
\begin{equation}\nonumber
0<U<(1-M^2)^{1/4}\sqrt{\frac{2 EI}{5\pi\rho_0 l b^3}}.
\end{equation}
\subsection{numerical study}
In this section, we aim to solve the characteristic equation $(I-T_2)^{-1}(1)(b)=0$ numerically. In other words, we want to find a numerical solution to the integral equation
\begin{equation}\label{eq:integral equation to be solved numerically}
f(y)+\mu\int_{-b}^{y}(y-s)^2[\mathcal{T}(f)](s)\, ds=1~~,~~|y|\leq b
\end{equation}
 satisfying the boundary condition
\begin{equation}\nonumber
f(b)=0.
\end{equation}
If a continuous solution to the integral equation exists satisfying the boundary condition, then using the Weierstrass theorem, the solution can be approximated using a polynomial. Therefore, we assume the following polynomial approximate solution
\begin{equation}\nonumber
P_n(y)=\sum_{i=0}^n a_i y^i
\end{equation}
with $P_n(b)=0$, where $n$ is the order of the polynomial. Due to \eqref{eq:integral equation to be solved numerically}, we also impose that $P_n(-b)=1$. After that, we define the error function $e(y)$ given by 
\begin{equation}\nonumber
e(y)=P_{n}(y)+\mu\int_{-b}^{y}(y-s)^2[\mathcal{T}(P_{n})](s)\, ds-1~~,~~|y|\leq b.
\end{equation}
To enable numerical computations, We try to satisfy the integral equation for a finite number, denote it by $m$, of points in the interval $[-b,b]$. Therefore, the interval $[-b,b]$ by a partitioning $-b\leq y_1\leq y_2\leq...\leq y_m\leq b$. Then, the coefficients $a_i$ and the parameter $\mu$ are obtained numerically by solving the following optimization problem.

\begin{equation}\label{eq:optimization problem}
 \min_{\mu,a_1,..,a_n}\max_{y_1,...,y_i,...,y_m} |e(y_i)|,~P_n(-b)=1,~P_n(b)=0.
\end{equation}
It is sufficient to solve the optimization problem for the case $b=1$ only. To illustrate this point, equation \eqref{eq:integral equation to be solved numerically} is non-dimensionalized as the following. Using the substitution $\tilde{y}=y/b$, equation \eqref{eq:integral equation to be solved numerically} can be written as 

\begin{equation}\nonumber
\tilde{f}(\tilde{y})+\mu b^3 \int_{-1}^{\tilde{y}}(\tilde{y}-s)^2[\tilde{\mathcal{T}}(\tilde{f})](s)\, ds=1~~,~~|\tilde{y}|\leq 1,
\end{equation}
where
$$\tilde{f}(\tilde{y})=f(b\tilde{y})$$
and
$$
\tilde{\mathcal{T}}(f(t))(x)=\frac{1}{\pi}\sqrt{\frac{1-x}{1+x}}\int_{-1}^{1}\sqrt{\frac{1+t}{1-t}}\frac{f(t)}{t-x}\,dt,~ |x|\leq 1.
$$
It can be seen from the transformed integral equation above that the value of $\mu$ obtained from the optimization problem \eqref{eq:optimization problem} by sitting $b=1$ is equal to $b^3\mu$. 
The optimization problem is solved using the FMINSEARCH tool in MATLAB. We ran the numerical computations for polynomials of orders from 2 to 8. The numerical computations show that the value of $b^3\mu$ approaches a value approximately equal to ${0.23}$ (see figure \eqref{fig:mu}) where the bound of $b^3\mu$ obtained in the analytical study is approximately ${0.13}$ and that is almost $56\%$ difference. Additionally, it can be seen from figure \eqref{fig:error}  that $\max_{y_i} |e(y_i)|$  decreases as the order of the polynomial increases with a minimum value approached approximately equal to $2.0\times 10^{-5}$ . It can be seen from figure \eqref{fig:solution profile} that the solution profile is captured starting from the second order approximation and as the polynomial order increases, the change in the approximate solution profile is very minimal. A minimization over the parameters $a_i$ only shows that $\max|e(y_i)|$ reaches its minimum value when $b^3\mu\approx 0.23$ (see figure \eqref{fig:minmu}) and that indicates that $\mu \approx 0.23/b^3$ is the minimum value such that the integral equation \eqref{eq:integral equation to be solved numerically} has a solution satisfying the boundary condition. Therefore, we can use $\mu \approx 0.23/b^3$ to obtain an approximate formula for the divergence speed which is given by
\begin{equation}\nonumber
U_{div}\approx(1-M^2)^{1/4}\sqrt{\frac{0.23 EI}{\rho_0 l b^3}}.
\end{equation}

 \begin{figure}
\includegraphics[width=5in ,keepaspectratio=true]{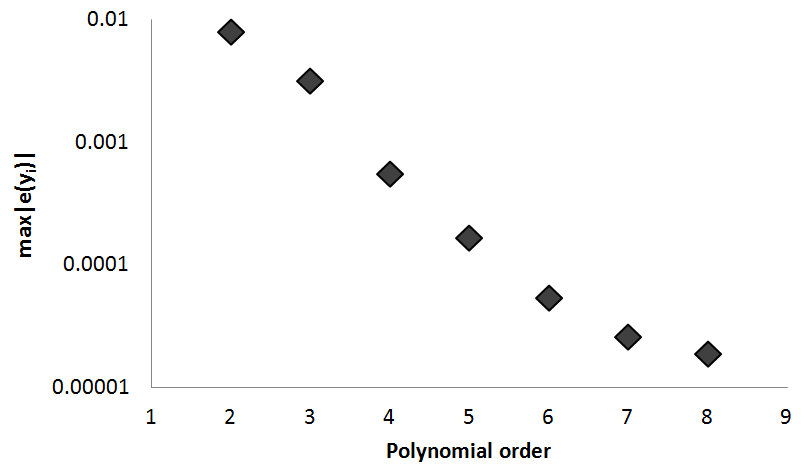}
\centering
\caption{ calculated value of $\max|e(y_i)|$ after solving the optimization problem \eqref{eq:optimization problem} for polynomial orders 2- 8.}
\label{fig:error}
\end{figure}

\begin{figure}
\includegraphics[width=5in ,keepaspectratio=true]{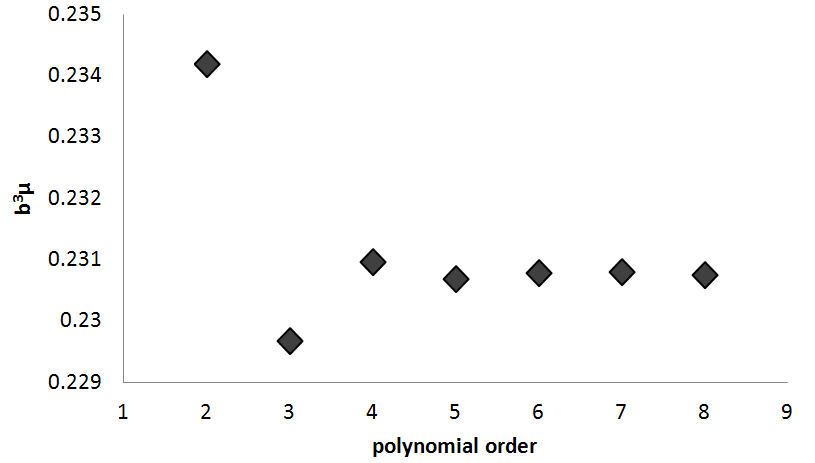}
\centering
\caption{calculated values of $b^3\mu$} for polynomial orders 2-8.
\label{fig:mu}
\end{figure}

 \begin{figure}
\includegraphics[width=5in ,keepaspectratio=true]{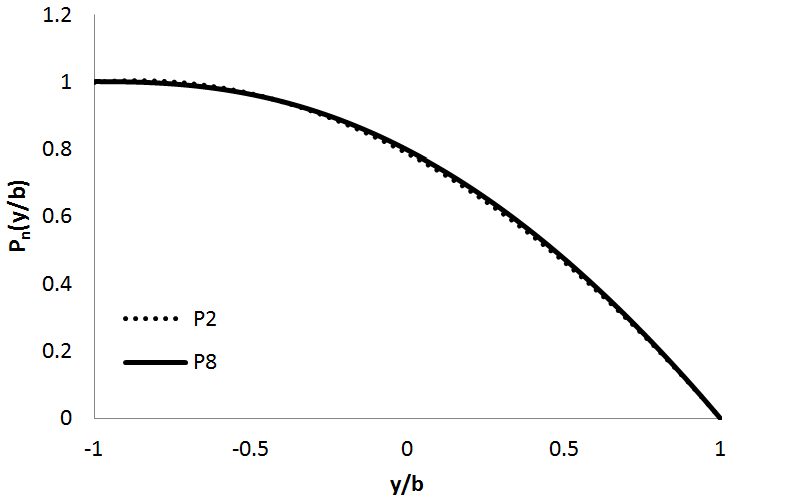}
\centering
\caption{approximate solutions profiles for  polynomials of order 2 and 8}
\label{fig:solution profile}
\end{figure}

 \begin{figure}
\includegraphics[width=5in ,keepaspectratio=true]{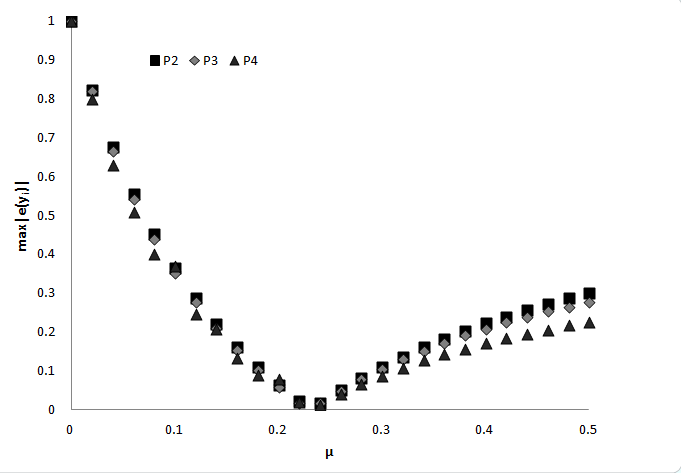}
\centering
\caption{ calculated value of $\max_{a_1,...,a_n}|e(y_i)|$ after optimizing the coefficients $a_i$ for different values of $\mu$ between 0 and 0.5 with b=1 and polynomial orders of 2, 3, and 4.}
\label{fig:minmu}
\end{figure}
 
\begin{remark} The framework presented in this paper can be used to study the static stability of thin plates with different boundary conditions (for example: clamped-clamped, pinned-pinned,etc). Following the same approach presented in the previous sections, the same characteristic equation \eqref{eq:characteristic equation} can be used but with different values of $\mathbf{P}$, $\mathbf{Q}$, and $\mathbf{u}$ that depend on the plate boundary conditions. Characteristic equations of plates with different boundary conditions  can then be obtained explicitly and analyzed analytically or numerically. In case of numerical treatment, the characteristic equations can be non-dimensionalized as illustrated in the previous discussion. Then, the approximate solutions to the non-dimensionalized characteristic equations can then used to obtain explicit formulas for the divergence speed.  
\end{remark} 

\begin{remark}[Comparison with some earlier works]
 In \cite{static stability 2}, the static stability of clamped-clamped plates in axial flow are studied analytically, numerically, and experimentally. The analytical treatment was based on the Galerkin method to approximate the profile of the plate deflection and the axial flow is assumed to be two dimensional and  potential. An analytical formula is then obtained and its accuracy is compared with numerical 3D simulations and experimentation. According to \cite{static stability 2}, the derived formula deviates from the numerical and experimental results by a factor of $2$, and it was proposed by the authors of that work to consider 3D models in order to derive accurate formulas of the divergence speed.
 
 In \cite{static instability 1}, the problem of axial potential flow over a pinned-pinned plate under tension is studied. The problem was formulated analytically resulting in an eigenvalue problem associated with solving an integral-differential equation. The eigenvalue problem was solved using the Galerkin method to approximate the eigenvalue problem and implementing numerical methods to solve the approximate  problem. The divergence speed is then plotted against different parameters of the aeroelastic problem to understand their effect on the divergence speed.
 
In comparison with the analytical and numerical treatments mentioned in the above works, we have the following comments. The analysis in this work covers subsonic compressible  flow and not only incompressible potential flow and therefore, the framework of our work is more general. Additionally, in contrast to these works which employ the Galerkin method, we retain the continuum model directly without discretizing the equations. Moreover, there is an advantage of using our framework to derive explicit formulas of the divergence speed. Even if the derived characteristic equations,  based on our framework, are solved numerically, explicit formulas of the divergence speed can be obtained as was the case in solving \eqref{eq:characteristic equation explicit} above.
 
 Although, we assume the flow equations to be two dimensional and the plate equation to be one dimensional, this framework can be a starting point and a basis for developing more sophisticated frameworks to analyze the static stabilities of thin structures in axial flows more accurately. 
\end{remark}

\section{Static stability of free-clamped piezoelectric flags}
There has been a recent trend of harvesting energy using piezoelectric flags by implementing them in axial flows. Therefore, it is important to predict the speed at which static instabilities of these flags may occur. Additionally, divergence speed of piezoelectric flags can be used as an upper bound of the flutter speed as in practice, divergence speed is larger than the flutter speed.
For piezoelectric plates, the  continuous models for the internal moment $\mathcal{M}$ and the charge transfer $Q$  are given by \cite{piezo electric continuous model}
\begin{equation}\nonumber
\mathcal{M}=EI h_{yy}-\mathcal{X}V,
\end{equation}
and
\begin{equation}\nonumber
Q=cV+\mathcal{X}h_{yy},
\end{equation}
where $c$ is the capacitance of the piezoelectric flag and $\mathcal{X}$ is a coupling term. Keeping the same assumption for the aerodynamic forces, the beam equation for the piezoelectric flag is 
\begin{equation}\nonumber
EI h_{yyyy}-\mathcal{X}V_{yy}=W(h_{y}).
\end{equation}
 If we assume zero charge transfer to the piezoelectric material and no energy dissipation inside the piezoelectric material, the momentum balance is reduced to be
 \begin{equation}\nonumber
 (EI+\frac{\mathcal{X}^2}{c}) h_{yyyy}=W(h_{y})
 \end{equation}
and for free-clamped piezoelectric flags, the boundary conditions are identical to the boundary conditions of the conventional free-clamped plates. Consequently, we have the free-pinned piezoelectric plates are statically unstable as in the case of conventional free-pinned plate. Additionally, the static stability flow velocity range for free-clamped piezoelectric flags is given by
\begin{equation}\label{eq:1}
0<U<(1-M^2)^{1/4}\sqrt{\frac{2(EI+\mathcal{X}^2/{c})}{5\pi\rho_0 l b^3}}.
\end{equation}
Moreover, the divergence speed of free-clamped piezoelectric flags is approximately given by
\begin{equation}\label{eq:2}
U_{div}\approx(1-M^2)^{1/4}\sqrt{\frac{0.23 (EI+\mathcal{X}^2/{c})}{\rho_0 l b^3}}.
\end{equation}

It is noticed from equations \eqref{eq:1} and \eqref{eq:2} that the piezoelectric coupling has a stabilizing effect as the equivalent bending stiffness $EI+\mathcal{X}^2/{c}$ increases with the piezoelectric coupling. Therefore, we propose implementing piezoelectric control to stabilize thin structures if the assumptions mentioned in the above discussion can be implemented physically.     
 
\section{Conclusion}
In this work, we analyze the static stability of plates with fixed trailing edges in subsonic axial air flow. We couple the deformation of the plate with the airflow using a singular integral equation, also known as the Possio integral equation, and then embed its steady state solution in the plate equation. Next, we verify the solvability of the static aeroelastic equations,while neglecting the boundary conditions, using tools from functional analysis. Then, we derive explicit formulas of the characteristic equations of free-clamped and free-pinned plates from which the divergence speed can be obtained. We show analytically that free-pinned plates are statically unstable as the divergence speed is zero. After that, we move to derive an analytic formula for the flow speeds  that correspond to static stability regions for free-clamped plates. We also resort to numerical computations to obtain an explicit formula for the divergence speed of free-clamped plates. Finally, we apply the obtained results on piezoelectric plates and we show that free-clamped piezoelectric plates are statically more stable than conventional free-clamped plates due to the piezoelectric coupling.


\begin{thebibliography}{}
\bibitem{static stability 2}
J. Adjiman, O. Doaré, and P. Moussou. \emph{ Buckling of a flat plate in a confined axial flow.} ASME 2015 Pressure Vessels and Piping Conference. American Society of Mechanical Engineers, 2015.

\bibitem{aeroelastic control4}
Y. Babbar, V. Suryakumar, T. Strganac, and A. Mangalam,  \emph{Measurement and Modeling of Nonlinear Aeroelastic Response under Gust}, 33rd AIAA Applied Aerodynamics Conference, 2015.

\bibitem{Balakrishnan solving possio equation}
A. V. Balakrishnan,\emph{ Possio integral equation of aeroelasticity theory}. Journal of Aerospace Engineering, 16.4, pp.~139--154, 2003.

\bibitem{Bal}
 A.~V.~ Balakrishnan, \emph{The Possio integral equation of aeroelasticity: a modern view. System modeling and optimization}, 15--22, IFIP Int. Fed. Inf. Process., 199, Springer, New York, 2006.
 
 \bibitem{iLiff paper}
A. V. Balakrishnan and K. Iliff, \emph{Continuum aeroelastic model for inviscid subsonic bending-torsion wing flutter}, Journal of Aerospace Engineering 20.3, pp.~152--164, 2007.

\bibitem{Balakrishnan Hindawi}
A. V. Balakrishnan and M. Shubov, \emph{Reduction of boundary value problem to Possio integral equation in theoretical aeroelasticity}, Journal of Applied Mathematics, 2008.


\bibitem{Balakrishnan book}
A.~V. Balakrishnan, Aeroelasticity: The Continuum Theory, Springer Science and Business Media, 2012.



\bibitem{Transonic Dip}
A.~V.~ Balakrishnan and A. Tuffaha,  \emph{The transonic dip in aeroelastic divergence speed -- An explicit formula}, Journal of the Franklin Institute, 349, pp.~59--73, 2012.



\bibitem{Axial flow conference paper}
A. V. Balakrishnan and A. Tuffaha, \emph{Aeroelastic flutter in axial flow-The continuum theory}, AIP Conference Proceedings-American Institute of Physics, Vol. 1493, No. 1, 2012.




\bibitem{static instability 1}
N.~Banichuka, J.~Jeronenb, P.~Neittaanma kib, T. ~Tuovinen, \emph{Static instability analysis for travelling membranes and plates interacting with axially moving ideal fluid}, Journal of Fluids and Structures 26, pp.~  274--291, 2010.

\bibitem{piezo electric continuous model}
O.~Doaré and S.~Michelin,\emph{ Piezoelectric coupling in energy-harvesting fluttering flexible plates: linear stability analysis and conversion efficiency.} Journal of Fluids and Structures, 27(8), pp.~1357--1375, 2011.

\bibitem{DSMD}
J.~Dunnmon, S.~Stanton, B.~Mann,~W. Dowell, \emph{ Power extraction from aeroelastic limit cycle oscillations}. Journal of Fluids and Structures 27, pp.~1182--1198, 2011.

\bibitem{aeroelastic control2}
C. Gao, G. Duan, and C. Jiang, \emph{Robust controller design for active flutter suppression of a two-dimensional airfoil}, Nonlinear dynamics and systems theory, pp.~ 287--299, 2009.

\bibitem{GWD}
S.~C.~Gibbs, I.~Wang, E.~Dowell,
\emph{Theory and experiment for flutter of a rectangular plate with a fixed leading edge in three-dimensional axial flow}, Journal of Fluids and Structures 34, pp.~ 68--83, 2012.

\bibitem{Hodges}
H. D. Hodges, and G. A. Pierce. Introduction to structural dynamics and aeroelasticity. Vol. 15. cambridge university press, 2011.


\bibitem{H}
L.~Huang, \emph{Flutter of cantilevered plates in axial flow}. Journal of Fluids and Structures 9 (2), 127--147, 1995.

\bibitem{aeroelastic control1}
D. Li, S. Guo, and J. Xiang, \emph{Aeroelastic dynamic response and control of an airfoil section with control surface non-linearities}, Journal of Sound and Vibration, 329.22, 4756--4771, 2010.

\bibitem{MD}
S.~ Michelin and O.~Doar\'e,
\emph{Energy harvesting efficiency of piezoelectric flags in axial flows}, J. Fluid Mech. 714, pp. 489-–504, 2013.


\bibitem{ground effect paper}
M. Serry and A. Tuffaha. Subsonic flow over a thin airfoil in ground effect. arXiv preprint arXiv:1702.04689, 2017.


 \bibitem{S}
 M~.A.~Shubov, \emph{Asymptotic representation for the eigenvalues of a non-selfadjoint operator governing the dynamics of an energy harvesting model}. Appl. Math. Optim. 73, no. 3, pp.~545--569, 2016. 
 


\bibitem{aeroelastic control3}
A.~Sutherland,~\emph{A Demonstration of Pitch-Plunge Flutter Suppression Using LQG Control},  International Congress of the Aeronautical Sciences, ICAS, 2010.

 \bibitem{Theodorsen}
T.~Theodorsen, \emph{General theory of aerodynamic instability and the mechanism of flutter}, Tech. Rep. 496, NACA, 1935.

\bibitem{Tricomi}
F.~Tricomi, Integral equations, Interscience Publishers Inc., 1957.











\end{thebibliography}
\end{document}